\documentclass[draft,12pt,onecolumn]{IEEEtran}
\usepackage{amsfonts,amsmath,amssymb}
\usepackage{graphicx}
\usepackage{subfigure}
\usepackage{multirow}

\def\beqa{\begin{eqnarray*}}
\def\eeqa{\end{eqnarray*}}
\def\be{\begin{eqnarray}}
\def\ee{\end{eqnarray}}

\def\1{{\bf 1}}

\def\u0{\underline{0}}
\def\ie{{\em i.e.},~}
\def\cf{{\em cf.},~}
\def\eg{{\em e.g.},~}

\newtheorem{lemma}{Lemma}[section]

\newlength{\defbaselineskip}
\newcommand{\setlinespacing}[1]%
           {\setlength{\baselineskip}{#1 \defbaselineskip}}

\newcommand{\eeq}{\end{equation}}

\newcommand{\beql}[1]{\begin{equation}\label{#1}}

\newcommand{\beq}{\begin{displaymath}}
\newcommand{\eeqno}{\end{displaymath}}

\begin{document}
\title{Golden-rule capacity allocation for distributed
delay management in peer-to-peer networks\thanks{This 
work was supported in part by NSF CNS grant 1152320.}}

\author{
\begin{tabular}{ccc}
George Kesidis and Guodong Pang\\
CSE, EE \& IME Depts,
The Pennsylvania  State University, University Park, PA, 16803\\
\{gik2,gup3\}@psu.edu  
\end{tabular}
}

\maketitle

\begin{abstract}
We describe  a distributed
framework for resources management in peer-to-peer networks
leading to ``golden-rule" reciprocity,  a kind of one-versus-rest
tit-for-tat, so that the delays experienced by any given peer's messages
in the rest of the network are proportional to those experienced
by others' messages at that peer.
\end{abstract}



\section{Introduction}

Consider a networked group of peers\footnote{Alternatively, 
a peer may be called a user or station, or, 
in the context of a network graph, a peer may
be called a node or vertex.}
 handling queries. In this paper, we consider
a cooperative group of peers operating,
\eg in support of a content-distribution system
\cite{IM08,Maggs08} or a social network
\cite{HKM13}.  Our framework is similar to the
cooperative Jackson-network 
model of \cite{TS13} 
whose aim is to minimize delays, however we assume
that the peers act in distributed/decentralized fashion
\cite{SVK12}.
We also assume that the aim of all peers is to achieve ``golden-rule"
reciprocity where the service capacity they allocate to foreign/transit
queries results in delay that is proportionate to that experienced
by exogenous (local) queries entering the network at that node.

This note is organized as follows. We motivate and set-up the problem
in Section \ref{setup-sec}. Our solution for golden-rule reciprocity
is given in Section \ref{soln-sec}. 
We discuss distributed computation of an eigenvector problem 
associated with the solution in Section \ref{distrd-sec}.
In Section \ref{numer-sec}, a numerical example is given and
we conclude with mention of future work in 
Section \ref{fw-sec}.

\section{Motivation and Problem Set-Up}\label{setup-sec}

We assume queries randomly forwarded so it's possible that
a query may be forwarded back to its originating peer.
Random forwarding in this manner may occur under ``onion"
routing \cite{TOR} for purposes of privacy protections. For example,
suppose that two peers $i$ and $j$ that are connected
(\cf $r_{i,j},r_{j,i}>0$)
share a private, symmetric cryptographic key $k_{i,j}$, \eg
\cite{Kurose-Ross}. Suppose when $i$ sends a query directly to $j$, and
$j$ forwards the query with $j$'s return address, \ie $j$ acts
as $i$'s proxy - so, it's possible that a peer later relaying
this query will send it to back to $i$ (as a query).

\subsection{Flow-balance for Randomized Forwarding}

A peer may have need to 
match a query to a cache of previously handled
queries \cite{HKM13} or using some other type of forwarding-information
base; thus queries may take a variable amount of time to process.
Each network peer $i\in\{1,2,...,N\}$  is assumed to have 
a limited ability to service queries,
herein modeled by i.i.d. exponentially distributed service
times with mean $1/\mu_i$. 

Peers not only consider their  query burden
(or those of their clients), for peer $i$
modeled as external arrival process with mean
rate $\lambda_{0,i}$, but also relay/resolve
queries on behalf of a much larger population
of peers.

Each peer can also control to which other peers 
it forwards queries which it cannot resolve.
Assume that peer $i$ can resolve a query 
with probability $r_{i,0}$. 

Suppose that the graph is connected so that
the network is stable, \ie
there exists a set of routing parameters
$r_{ij}\geq 0$ such that the
solution of flow-balance equations
\beqa
\Lambda_i = \lambda_{0,i} + \sum_{j=1}^N \Lambda_jr_{j,i}, &\mbox{with} &
 1 = \sum_{j=0}^N r_{i,j},
\eeqa
such that $\Lambda_i < \mu_i$ for all $i$,
where $r_{ij}=0$ if $i$ and $j$ are not directly connected. 
Defining the $N\times N$ routing matrix $R=[r_{i,j}]$ 
for $1\leq i,j\leq N$ and
identity matrix $I$, 
we can collectively write the the flow-balance equations as
\beqa
\underline{\Lambda}' (I- R)   =  \underline{\lambda}_0'
& \Rightarrow & \underline{\Lambda}'  =   \underline{\lambda}_0' B,
~~\mbox{where}\\
B = [b_{i,j}]  :=  ( I- R)^{-1},  & \mbox{\ie} &
\Lambda_i  =  \sum_{j=1}^N b_{j,i}\lambda_{0,j} ~~\forall i.
\eeqa
We assume that $ R$ is strictly sub-stochastic 
(not all $r_{i,0}$ are zero) so that $ I- R$ is non-singular.
Note that since $ R\geq 0$ (a non-negative matrix with
all entries $r_{i,j}\geq 0$) then also
\be
B & = & (I-R)^{-1} ~=~ I + R+ R^2+R^3 + ... ~\geq ~0.  \label{B-is-nonnegative}
\ee
Note that $b_{i,i}\geq 1$, \ie $\Lambda_i  \geq  \lambda_{0,i}$
(again, some queries may return (as queries) to their originating peer). 

\subsection{Queues for Local and Foreign Queries at each Peer}

Each peer $i$ has two queues, one for its local load 
\beqa
b_{i,i}\lambda_{0,i} 
\eeqa
and one for its foreign load $\Lambda_i - b_{i,i}\lambda_{0,i}$,
respectively served at rates
\be
\mu_{0,i}>b_{i,i}\lambda_{0,i} & \mbox{and} & 
\mu_i-\mu_{0,i} > \Lambda_i-b_{i,i}\lambda_{0,i}, ~~~\forall i,
\label{stable-queues}
\ee 
where we have made an implicit assumption of queue stability
on the available service
capacities $\mu_i,\mu_{0,i}$. 
These two queues are effectively multiplexed before forwarding  
to peers in $\{0,1,2,...,N\}$  because, for simplicity of 
exposition herein, we 
assume the same forwarding probabilities 
$r$ for both queues of a given peer. 

Note that this can be easily generalized 
to different forwarding probabilities for the local and foreign load
of a given peer, with good reason from a modeling perspective; 
\eg a peer may be able to resolve 
(forward to peer $0$) a local query with different probability
than a foreign one.

\subsection{The Jackson Network Model}

Assuming Poisson external arrivals 
(again, at rates $\underline{\lambda}_0$) gives a Jackson network model
with $2N$ queues.  The mean number of local queries queued at peer $i$  is
\beqa
L_{i,i} & = & \frac{b_{i,i}\lambda_{0,i}}{\mu_{0,i}-b_{i,i}\lambda_{0,i}}
\eeqa
The mean number of queries in $j$'s foreign-traffic queue is
\beqa
L_j &:= &\frac{\Lambda_j-b_{j,j}\lambda_{0,j}}
{\mu_j-\mu_{0,j} - (\Lambda_j-b_{j,j}\lambda_{0,j})}
\eeqa
The mean number of $i$'s jobs in $j$'s foreign-traffic queue is
\beqa
L_{i,j} &:= & 
\frac{b_{i,j}\lambda_{0,i}} 
{\Lambda_j-b_{j,j}\lambda_{0,j}} \cdot
\frac{\Lambda_j-b_{j,j}\lambda_{0,j}}
{\mu_j-\mu_{0,j} - (\Lambda_j-b_{j,j}\lambda_{0,j})}\\
& = & 
\frac{b_{i,j}\lambda_{0,i}} 
{\mu_j-\mu_{0,j} - (\Lambda_j-b_{j,j}\lambda_{0,j})}.
\eeqa
The mean {\em delay} (disutility) of foreign traffic at peer $i$ is
\beqa
\frac{1}{\mu_i-\mu_{0,i} - (\Lambda_i-b_{i,i}\lambda_{0,i})}.
\eeqa
Using Little's theorem \cite{Wolff89},
we can write the {\em disutility of}  (delay, cost to) peer $i$ 
as a combination of the total mean delay of peer $i$'s
traffic and the mean local delay of foreign traffic at peer $i$
(of all other traffic),
the latter modified by an 
``altruism" (cooperation) factor  $\alpha_i>0$:
\beqa
C_i & = & 
\frac{1}{\lambda_{0,i} }
\sum_{j=1}^N 
L_{i,j}  + 
 \alpha_i\frac{1}{\mu_i-\mu_{0,i} - (\Lambda_i-b_{i,i}\lambda_{0,i})}\\
& = & 
\frac{b_{i,i}}{\mu_{0,i}-b_{i,i}\lambda_{0,i}}
+\sum_{j\not = i} \frac{b_{i,j}} {\mu_j-\mu_{0,j} - (\Lambda_j-
b_{j,j}\lambda_{0,j})}\\
& & ~~
 +\alpha_i\frac{1}{\mu_i-\mu_{0,i} - (\Lambda_i-b_{i,i}\lambda_{0,i})}.
\eeqa

\section{Optimal Service Allocations for Golden-Rule
Reciprocity}\label{soln-sec}

Fixing $\mu_j$ for all $j$, 
note how $C_i$ also depends on $\mu_{0,j}$ for all $j\not = i$,
subject to (\ref{stable-queues}).
Solving the first-order Nash equilibrium conditions
with $\mu_{0,i}$ the play actions:
\beqa
\partial C_i / \partial \mu_{0,i} & = & 0 ~~\forall i,
\eeqa
leads to the unique  solution
\be
\mu_{0,i}^*  & = &  
\frac{\sqrt{b_{i,i}}}{\sqrt{b_{i,i}}+\sqrt{\alpha_i}} (\mu_i-\Lambda_i)
+ b_{i,i}\lambda_{0,i} \label{mu0-equ}\\
\Rightarrow ~~
\mu_i-\mu_{0,i}^*  & = &  \frac{\sqrt{\alpha_i}}{\sqrt{b_{i,i}}+\sqrt{\alpha_i}} 
(\mu_i-\Lambda_i)
+\Lambda_i- b_{i,i}\lambda_{0,i}. \nonumber
\ee
Since $C_i$ is concave in $\mu_{0,i}$ for all $i$,
$\underline{\mu}_0^*$ is the unique Nash  equilibrium \cite{Rosen65}.
Thus, (\ref{stable-queues}) holds when 
all $\alpha_i \in(0,\infty)$,
with mean delay diverging in the foreign queue as $\alpha\downarrow 0$
($i$ increasingly selfish), and 
mean delay diverging in the local queue as $\alpha\rightarrow \infty$
($i$ increasingly selfless).



\subsection{Golden-Rule Reciprocity}

Rather than exploring the efficiency of Nash equilibria
according to a global ``social welfare" function (as in \eg \cite{Kelly98c}),
suppose we simply want a Nash equilibrium such that: for all peers $i$,
the mean delay experienced by $i$'s queries in the rest of the network
is proportional to the the mean delay in $i$'s foreign-traffic queue.
That is, using $\underline{\mu}_0=\underline{\mu}_0^*$ for all $i$,  
there is a constant $\kappa$ (constant across all peers) such that:
\be
\forall i, ~~\sum_{j\not = i} 
\frac{1+\sqrt{b_{j,j}/\alpha_j}}{\mu_j - \Lambda_j} b_{i,j} & = & 
 \kappa \frac{1+\sqrt{b_{i,i}/\alpha_i}}{\mu_i - \Lambda_i}.
\label{ev-equ}
\ee
Recall $B:=(I-R)^{-1}\geq 0$  by (\ref{B-is-nonnegative}), and define
a zero-diagonal version of $B$,
\beqa
\tilde{B} & := & B-{\sf diag}\{B\} ~\geq ~0,
\eeqa
where ${\sf diag}\{B\}$ is the diagonal matrix 
with entries $b_{i,i}$ (\ie those of $B$).

\begin{lemma}
If a square matrix $R\geq 0$ is irreducible, then the following
matrix is also non-negative and irreducible:
\beqa
(I-R)^{-1}-{\sf diag}\{(I-R)^{-1}\}
& =: & 
\tilde{B} .
\eeqa
\end{lemma}
\begin{proof}
Clearly $\tilde{B}\geq 0$ by (\ref{B-is-nonnegative}), \ie since
$B:=(I-R)^{-1}\geq 0$. Irreducibility is a property of
the off-diagonal elements of a matrix that are zero. 
$\tilde{B}$ is irreducible because 
\beqa
\forall i\not = j,~~
\tilde{b}_{i,j} & = &  (R+R^2+R^3+...)_{i,j} ~\geq ~ r_{i,j},
\eeqa
and $R$ is irreducible by hypothesis.
\end{proof}
Now since $\tilde{B}\geq 0$ is irreducible,
by the Perron-Frobenius theorem \cite{horn0} (see also Theorem 5, p. 9,
of \cite{Noutsos}), there
is a positive right-eigenvector $\underline{v}$  corresponding
to a maximal eigenvalue ($\kappa > 0$) of $\tilde{B}$, \ie
$\tilde{B}\underline{v}=\kappa\underline{v}$ as (\ref{ev-equ}).

Clearly then, to achieve golden-rule reciprocity 
(\ref{ev-equ}), each peer $i$ should set
their altruism parameter such that
\be
v_i & =&  \frac{1+\sqrt{b_{i,i}/\alpha_i}}{\mu_i - \Lambda_i} \nonumber \\
 \Rightarrow~~ \forall i, ~~\alpha_i   & = &  
b_{i,i}(v_i(\mu_i - \Lambda_i)-1)^{-2},
\label{golden-rule-condition}
\ee
which is feasible when $\forall i, ~~v_i(\mu_i - \Lambda_i)  >  1$, \ie
\be\label{feasible-condition}
\forall i, ~~\mu_i  & > & v_i^{-1}  + \Lambda_i.
\ee

\subsection{Summary}

In our Jackson network setting, given the external demand
$\underline{\lambda}_0$ 
and assuming strictly sub-stochastic $R$, and hence 
$\tilde{B}\geq 0$, 
are irreducible: 
\begin{enumerate}
\item  the total loads $\underline{\Lambda}$ are naturally
computed in distributed fashion by flow-balance,
\item peers compute the strictly
positive eigenvector $\underline{v}>0$ of $\tilde{B}$ 
in distributed fashion (\cf discussion
below),
\item each peer allocates sufficient service capacity $\mu$ to achieve 
(\ref{feasible-condition}) 
(alternatively, if service capacity $\mu_i$ of peer $i$ is limited,
then their demand $\lambda_{0,i}$ is reduced/thinned to achieve 
(\ref{feasible-condition})),
\item each peer sets their altruism parameter $\alpha$  as
(\ref{golden-rule-condition}), and
\item each peer sets their service allocations for their foreign
and local queues according to (\ref{mu0-equ}).
\end{enumerate}

\section{Discussion: 
Achieving Golden-Rule Reciprocity in Distributed/Decentralized 
Fashion}\label{distrd-sec}

The routing information $R$ could simply be shared among all the peers
in the manner of link-state routing based on Dijkstra's algorithm
\cite{Kurose-Ross}.
Alternatively,
Section 2 of \cite{KM04} describes a method for {\em distributed}
computation of the principle eigenvector of a routing/graph matrix
based on orthogonal iteration and 
fast-mixing random walks \cite{KDG03}.
Though the eigenvectors of $R$ and $B=(I-R)^{-1}$ coincide, those
of $R$ and $\tilde{B}=B-\mbox{diag}\{B\}$ necessarily do not.
To accommodate $\tilde{B}$, we 
now specify a modification of the orthogonal iteration 
stated as Algorithm 1 of \cite{KM04} (where their $Q=\underline{v}>0$
and their $A=\tilde{B}$).
Let $\Omega(\underline{v})$ be the normalization of $\underline{v}$.

\begin{itemize} 
\item[0:] initialize $B_0 = I$ and $\underline{v}_0$
 as a random, non-zero, and non-negative vector.
\item[$k$:] For step $k>0$, do:
\begin{itemize}
\item[$k.1$:] $B_k = I + B_{k-1}R$
\item[$k.2$:]
$v_{k}= \Omega((B_k - \mbox{diag}{B_k})v_{k-1})$
\item[$k.3$:]
If the change in $\underline{v}$ and $B$ are not negligible,
go to step $k+1$, else stop.
\end{itemize}
\end{itemize}
 
Note that $B_k$ converges to $B=(I-R)^{-1}$ (again, $R$ is strictly
sub-stochastic by assumption).
The approach used for decentralized (ortho)normalization of Section 2.2
of \cite{KM04}, again based on fast-mixing random walks \cite{KDG03},
can also be used to update the $B$ matrix in step $k.1$.

\section{Numerical Example}\label{numer-sec}

For a simple example network of three 
peers\footnote{Computed at http://www.bluebit.gr/matrix-calculator}
with routing probabilities
\beqa
R =\frac{1}{6} 
\left[ \begin{array}{ccc} 
0 & 2 & 3\\
2 & 0 & 3 \\
3 & 1 & 0
\end{array}\right] 
\Rightarrow B=
\left[ \begin{array}{ccc} 
2.062 & 0.937 & 1.500\\
1.312 & 1.687 & 1.500\\
1.250 & 0.750 & 2.000
\end{array}\right],
\eeqa
and maximum eigenvalue of $\tilde{B}:= B-\mbox{diag}(B)$
is $\kappa = 2.366$ with associated positive right-eigenvector
\beqa
\underline{v} & = & [0.576 ~~ 0.641 ~~ 0.507]'.
\eeqa
If exogenous loads are $\underline{\lambda}_0 =
[1~~ 2 ~~1]'$, then by flow-balance, total
loads are 
\beqa
\underline{\Lambda}' :=
\underline{\lambda}_0'B :=
\underline{\lambda}_0' (I-R)^{-1}=
[5.9 ~~5.061 ~~ 6.5].
\eeqa
Taking the service rates $\mu$ sufficiently
large to 
satisfy (\ref{feasible-condition}),  \eg as
\beqa
\underline{\mu}  &=&   
[8 ~~7~~ 9]' > 
[7.636 ~~6.621 ~~ 8.472]'  = \underline{v}^{-1}+\underline{\Lambda},
\eeqa
the golden-rule cooperation parameters
$\alpha$ are computed using (\ref{golden-rule-condition}):
\beqa
\underline{\alpha} & = & 
[46.9 ~~28.6~~ 28.0]'  
\eeqa
Finally, the bandwidth allocations for the local 
and foreign queues
under golden-rule cooperation according to (\ref{mu0-equ})
are, respectively,
\beqa
\underline{\mu}_0 & = &[2.43 ~~3.75~~ 2.32 ]'  \\
\Rightarrow~~
\underline{\mu}-\underline{\mu}_0 & = &[5.57 ~~3.25 ~~ 6.68]'  
\eeqa

\section{Future Work}\label{fw-sec}

Variations of the  disutilities $C_i$ given above can
be studied, possibly
leading to similar golden-rule reciprocity solution frameworks.

Also, it's well known how  peer-to-peer
systems may experience degraded performance due to various
types of selfish behavior, see \eg
\cite{Jin13,GK13}.
Future work on this problem includes consideration of
how to police and
 mitigate defectors from golden-rule reciprocity, 
particularly in the process through which routing information
$R$ is shared among the peers and the maximal, positive
eigenvalue $\kappa$ of is computed. That is, some peers 
may want their traffic to 
be treated {\em better} by their peers than they themselves
treat their peers' traffic.  


\end{document}